\documentclass[journal,onecolumn]{IEEEtran}
\usepackage[T1]{fontenc}
\usepackage{amsmath,amssymb,amsthm,mathtools}
\usepackage[hidelinks]{hyperref}
\newtheorem{theorem}{Theorem}[section]
\newtheorem{proposition}[theorem]{Proposition}

\newtheorem{corollary}[theorem]{Corollary}
\theoremstyle{remark}

\newcommand{\R}{\mathbb{R}}

\newcommand{\SU}{\mathrm{SU}}
\newcommand{\SO}{\mathrm{SO}}
\newcommand{\Gr}{\mathrm{Gr}}
\newcommand{\Sym}{\mathrm{Sym}}
\newcommand{\Tr}{\operatorname{Tr}}
\newcommand{\Cl}{\mathcal{C}}

\newcommand{\Uclosed}{\mathcal{U}_{18}}
\newcommand{\norm}[1]{\left\lVert #1\right\rVert}
\begin{document}
\title{Deep Holes in the Clifford Hierarchy}
\author{Ian~Teixeira, David~Meyer
\thanks{Both authors are with the Department of Mathematics, University of California, San Diego, CA 92093}
}
\maketitle
\begin{abstract}
We determine the covering radius of the topological closure of the
single-qubit Clifford hierarchy in $\SU(2)\cong S^3$.  This closure is a
union of $18$ great circles --- the Clifford--Pauli circles --- and we
prove that its covering radius is $\arccos\sqrt{5/6}$.  The extremal
points, which we call \emph{deep holes}, form a single orbit of size
$192$ under left and right multiplication by Clifford gates, and are
described in closed form.  Equivalently, the minimum over one-qubit
unitaries of the all-level Clifford fidelity is $5/6$. The proof rests on two structures attached to the
configuration of $18$ planes in $\R^4$: their centered rank-two
projectors form an orthonormal basis of the irreducible $\SO(4)$-module
$\Sym_0(4)$, and the projection profile of a
unit quaternion is exactly its image under the double cover
$\SU(2)\to\SO(3)$.  These reduce the covering problem to a minimax
statement for the $\ell^\infty$-norm on $\SO(3)$ which we solve exactly,
classifying its equality cases.
\end{abstract}
\section{Introduction}
The Clifford hierarchy is the recursively defined sequence
$\Cl_1\subset\Cl_2\subset\cdots$ of quantum gates obtained by requiring
conjugation of Pauli operators to descend one level.  Its first two levels
are the Pauli and Clifford groups, while higher levels contain familiar
non-Clifford gates that can be implemented by gate-teleportation
constructions \cite{teleportationGottesmanChuang}.  The hierarchy also
appears naturally in restrictions on fault-tolerant logical gates:
Eastin--Knill rules out a universal transversal gate set for a nontrivial
quantum code \cite{eastinKnill}, Bravyi--K\"onig constrain
locality-preserving logical gates in topological stabilizer codes
\cite{BravyiPRL}, and disjointness gives related hierarchy-level bounds for
transversal and constant-depth logical gates on stabilizer codes
\cite{disjointnessStabilizerCodes}.  These results make the geometry of the
hierarchy itself a natural object to study.
For one qubit that geometry is unusually rigid.  Zeng, Chen, and Chuang
proved that every single-qubit hierarchy gate is semi-Clifford
\cite{semiClifford}.  Subsequent work has clarified the structure of the
third and higher levels from several directions
\cite{BeigiShor2010,UnWeylingtheCliffordHierarchy,Anderson_2024}, and the
recent one-qudit classification gives a complete normal form in
prime dimension \cite{deSilva2025clifford}.  Together with the
classification of diagonal hierarchy gates \cite{CuiGottesmanKrishna}, the
single-qubit result implies that, as the level grows, the hierarchy becomes
dense not in all of $\SU(2)$ but in a finite union of one-parameter Clifford
translates.  Under $\SU(2)\cong S^3$, these translates are great circles.
Our first observation is that there are exactly eighteen such circles.  We
write their union as
\[
    \Uclosed=\bigcup_{r=1}^{18}(S^3\cap\Pi_r),
    \qquad \Pi_r\in\Gr(2,4).
\]
The main result is the following.
\begin{theorem}[Covering radius and deep holes]\label{thm:main}
With the round metric on $S^3\cong\SU(2)$,
\[
    \rho(\Uclosed)=\arccos\sqrt{\frac56}.
\]
The points attaining this distance form a single orbit under left and
right multiplication by the determinant-one single-qubit Clifford group.
There are exactly $192$ such points in $S^3$, hence $96$ projective
single-qubit gates modulo global phase.
\end{theorem}
Recent work on testing membership in the Clifford hierarchy uses the
normalized-trace overlap to define the degree-$k$ Clifford fidelity
\cite{Bao2026HierarchyTesting}.  In the present one-qubit setting it is
natural to consider the all-level quantity
\begin{equation}\label{eq:all-level-fidelity}
    F_\infty(U)
      :=\sup_{k\geq1}\ \max_{V\in\Cl_k}
      \left|\frac12\Tr(V^\dagger U)\right|^2.
\end{equation}
Under the quaternion identification $\SU(2)\cong S^3$, the normalized
trace inner product is the Euclidean inner product on $\R^4$.  Since the
union of the hierarchy levels is dense in $\Uclosed$, the quantity in
\eqref{eq:all-level-fidelity} is exactly the largest squared projection
of the corresponding unit quaternion onto one of the eighteen planes.
Thus Theorem~\ref{thm:main} is equivalently the sharp fidelity statement
\[
    \min_{U\in\SU(2)}F_\infty(U)=\frac56,
\]
and its deep holes are precisely the equality cases.
This theorem computes the greatest distance from a single-qubit unitary to the closed hierarchy.  Because low levels of the hierarchy are important for fault-tolerant schemes based on gate teleportation and are also singled out by general restrictions on protected logical gates \cite{teleportationGottesmanChuang,eastinKnill,disjointnessStabilizerCodes,BravyiPRL}, this gives a precise sense in which a unitary can be maximally far from the best-understood single-qubit hierarchy operations.  It complements exact and approximate synthesis results over Clifford-type gate sets, such as Clifford$+T$, by identifying the unitaries farthest from all Clifford translates of Pauli-axis rotations \cite{KliuchnikovMaslovMosca2013,RossSelinger2016}.
There is also a classical coding-theoretic side to the problem.  Covering
questions for highly symmetric finite configurations are closely connected
with Delsarte linear programming and association schemes
\cite{Delsarte1973,ConwaySloane,BCN1989}.  Finite configurations of
subspaces bring in Grassmannian designs and fusion frames
\cite{GrassmannianDesignsBachocCoulangeonNebe2002,
FramesCasazzaKutyniok2004,FramesCasazzaKutyniokLi2008,Bachoc2006}.
Their pair geometry has only three relations.  These features explain much
of the symmetry of the configuration and place it naturally near
association-scheme methods, but the sharp covering theorem does not
require linear programming or Bose--Mesner machinery.
Instead, the proof has two simple pieces.  First, after centering the
eighteen rank-two projectors, one obtains nine orthonormal directions and
their negatives in $\Sym_0(4)$.  Second, the projection profile of a point
of $S^3$ is the usual quaternionic parametrization of $\SO(3)$.  Thus the
nonlinear covering problem on $S^3$ becomes the minimax problem
\[
    \min_{R\in\SO(3)}\max_{i,j}|R_{ij}|.
\]
\section{The closed single-qubit hierarchy is eighteen circles}
\label{sec:circles}
We use determinant-one representatives throughout.  Let $X,Y,Z$ be the
usual Pauli matrices and set
\[
    \mathbf i=-iX,\qquad
    \mathbf j=-iY,\qquad
    \mathbf k=-iZ.
\]
Then
\[
    \mathbf i^2=\mathbf j^2=\mathbf k^2=-I,
    \qquad
    \mathbf i\mathbf j=\mathbf k,
    \quad
    \mathbf j\mathbf k=\mathbf i,
    \quad
    \mathbf k\mathbf i=\mathbf j,
\]
so that
\[
    \SU(2)
    =\{aI+b\mathbf i+c\mathbf j+d\mathbf k:
      a^2+b^2+c^2+d^2=1\}
    \cong S^3.
\]
For $P\in\{\mathbf i,\mathbf j,\mathbf k\}$ define the Pauli rotation
circle
\[
    T_P:=\{e^{\theta P}:\theta\in\R\}
        =S^3\cap\operatorname{span}_{\R}\{I,P\}.
\]
Let $\Cl_2\leq\SU(2)$ denote the determinant-one Clifford group.  It is
the binary octahedral group $2O$ of order $48$.  The semi-Clifford theorem
for one qubit says that any hierarchy gate can be written
\[
    U=C_1DC_2,
    \qquad C_1,C_2\in\Cl_2,
\]
with $D$ diagonal \cite{semiClifford}.  In determinant-one form
$D=e^{\theta\mathbf k}$.  The diagonal classification \cite{CuiGottesmanKrishna} gives
\[
    e^{\theta\mathbf k}\in\Cl_k
    \quad\Longleftrightarrow\quad
    \theta\in\frac{2\pi}{2^{k+1}}\mathbb Z.
\]
Hence the union over
$k$ of diagonal hierarchy gates is dense in $T_{\mathbf k}$.
Absorbing the right Clifford into a conjugation gives
\[
    C_1e^{\theta\mathbf k}C_2
      =(C_1C_2)
       e^{\theta(C_2^{-1}\mathbf kC_2)}.
\]
Since Clifford conjugation permutes the signed Pauli directions, the
closure of the entire single-qubit hierarchy is therefore
\begin{equation}\label{eq:closed-hierarchy}
    \overline{\bigcup_{k\geq1}\Cl_k}
    =\bigcup_{C\in\Cl_2}
      \bigcup_{P\in\{\mathbf i,\mathbf j,\mathbf k\}} CT_P.
\end{equation}
For completeness, the reverse inclusion implicit here uses the
fact that each hierarchy level $\Cl_k$ with $k\geq2$ is preserved by left
and right Clifford multiplication.  Dyadic
rotations are dense in each $T_P$, so every point of every circle on the
right-hand side of \eqref{eq:closed-hierarchy} is a limit of hierarchy
gates.
For fixed $P$, two Clifford elements give the same left coset circle
precisely when they differ by an element of $\Cl_2\cap T_P$.  This
intersection is cyclic of order $8$, so there are
\[
    \frac{|\Cl_2|}{|\Cl_2\cap T_P|}=\frac{48}{8}=6
\]
distinct circles for each $P$.  Cosets belonging to two distinct Pauli
rotation subgroups cannot coincide: if $CT_P=C'T_Q$, translating a common
point to the identity would give $T_P=T_Q$.  Thus the three Pauli
directions give exactly
\[
    3\cdot6=18
\]
great circles.
Every such circle is the unit section of a real two-plane:
\[
    CT_P
    =S^3\cap\operatorname{span}_{\R}\{C,CP\}.
\]
We denote the resulting planes by
$\Pi_1,\ldots,\Pi_{18}\subset\R^4$ and their orthogonal projectors by
$P_1,\ldots,P_{18}$.
\section{The eighteen centered projectors}
\label{sec:planes}
The geometry becomes transparent after centering the projectors.  An
explicit ordering of the planes, given in Appendix~\ref{app:Q}, has the
following form.
\begin{proposition}[Centered-projector model]\label{prop:Qbasis}
There are symmetric traceless matrices
$Q_1,\ldots,Q_9\in\Sym_0(4)$ such that
\begin{equation}\label{eq:PfromQ}
    P_i=\frac12I_4+Q_i,
    \qquad
    P_{i+9}=\frac12I_4-Q_i,
    \qquad 1\leq i\leq9,
\end{equation}
and
\begin{equation}\label{eq:Qrelations}
    Q_i^2=\frac14I_4,
    \qquad
    \Tr(Q_iQ_j)=\delta_{ij}.
\end{equation}
Consequently $Q_1,\ldots,Q_9$ form an orthonormal basis of the
nine-dimensional space $\Sym_0(4)$, and the eighteen centered projectors
are exactly
\[
    \{P_r-\tfrac12I_4:1\leq r\leq18\}
       =\{\pm Q_1,\ldots,\pm Q_9\}.
\]
\end{proposition}
\begin{proof}
The matrices are displayed in Appendix~\ref{app:Q}.  The identities in
\eqref{eq:Qrelations} follow immediately by multiplication.  Since
$Q_i^2=I_4/4$ and $\Tr Q_i=0$, the matrices
$I_4/2\pm Q_i$ are complementary rank-two orthogonal projectors.  Direct
left multiplication of the six coordinate planes by the three
Hadamard-type Clifford elements
\[
    \frac{\mathbf i+\mathbf j}{\sqrt2},\qquad
    \frac{\mathbf i+\mathbf k}{\sqrt2},\qquad
    \frac{\mathbf j+\mathbf k}{\sqrt2}
\]
gives precisely these eighteen projectors.
\end{proof}
This orthonormal description immediately gives the low-degree moment
identities of the configuration.
\begin{corollary}[Grassmannian design]\label{cor:design}
The eighteen planes form a Grassmannian $3$-design in $\Gr(2,4)$.
In particular,
\begin{equation}\label{eq:firstmoment}
    \frac1{18}\sum_{r=1}^{18}P_r=\frac12I_4
\end{equation}
and, for every $A\in\Sym_0(4)$,
\begin{equation}\label{eq:secondmoment}
    \frac1{18}\sum_{r=1}^{18}
       \Tr\!\left(A(P_r-\tfrac12I_4)\right)
       (P_r-\tfrac12I_4)
    =\frac19A.
\end{equation}
\end{corollary}
\begin{proof}
Write
\[
    Q=P-\frac12I_4
\]
for the centered projector of a two-plane.  By
Proposition~\ref{prop:Qbasis}, the centered projectors of our eighteen
planes are simply
\[
    \{\pm Q_1,\ldots,\pm Q_9\},
\]
where $Q_1,\ldots,Q_9$ is an orthonormal basis of the
nine-dimensional space $\Sym_0(4)$.  Thus, after centering, our
configuration is just an orthonormal basis together with its negatives.
The first moment is immediate: the centered projectors cancel in
opposite pairs, so
\[
    \frac1{18}\sum_{r=1}^{18}P_r
    =\frac12I_4
      +\frac1{18}\sum_{i=1}^9(Q_i-Q_i)
    =\frac12I_4.
\]
The second moment is equally simple.  Using the Frobenius inner product
$\langle A,B\rangle=\Tr(AB)$ and the orthonormality of the $Q_i$,
\[
\begin{aligned}
    \frac1{18}\sum_{r=1}^{18}
       \langle A,P_r-\tfrac12I_4\rangle
       (P_r-\tfrac12I_4)
    &=
    \frac1{18}\sum_{i=1}^9
       \left(
       \langle A,Q_i\rangle Q_i
       +\langle A,-Q_i\rangle(-Q_i)
       \right)\\
    &=
    \frac19\sum_{i=1}^9
       \langle A,Q_i\rangle Q_i
     =\frac19A.
\end{aligned}
\]
It remains to check that these are the corresponding Haar moments on
$\Gr(2,4)$.  Let $P$ be a uniformly random rank-two projector and set
$Q=P-I_4/2$.  Then $Q\in\Sym_0(4)$ and
\[
    \|Q\|_F^2=1.
\]
The Haar distribution is invariant under $P\mapsto UPU^T$ for
$U\in O(4)$.  Hence its covariance operator
\[
    A\longmapsto
    \mathbb E\!\left[\langle A,Q\rangle Q\right]
\]
commutes with the $O(4)$-action on $\Sym_0(4)$.  This representation is
irreducible, so the covariance must be a scalar multiple of the
identity.  Since $\dim\Sym_0(4)=9$ and $\|Q\|_F^2=1$, that scalar is
$1/9$.  Thus the finite configuration and Haar measure have the same
first and second moments.
Finally, orthogonal complementation sends
\[
    P\longmapsto I_4-P,
    \qquad
    Q\longmapsto -Q.
\]
Both Haar measure and our finite configuration are invariant under this
operation.  Consequently all odd centered moments vanish for both, in
particular the third moment.  The two distributions therefore agree in
centered moments of degrees $0,1,2,$ and $3$.  Since
$P=I_4/2+Q$, they agree on every polynomial of degree at most $3$ in
the projector entries.  This is exactly the Grassmannian
$3$-design property.
\end{proof}
For completeness, a direct calculation from the matrices in
Appendix~\ref{app:Q} gives the pair geometry.  Among the
$\binom{18}{2}=153$ unordered pairs of planes, there are $9$ orthogonal
pairs, $72$ pairs with principal angles $(0,\pi/2)$, and $72$ isoclinic
pairs with principal angles $(\pi/4,\pi/4)$.  Thus each plane has one
orthogonal complement, eight planes meeting it in a line, and eight
planes isoclinic to it at angle $\pi/4$.  These counts are not needed for
the covering-radius calculation.
\section{From projection coordinates to \texorpdfstring{$\SO(3)$}{SO(3)}}
\label{sec:SO3-reduction}
For $u\in S^3\subset\R^4$, define
\[
    x_r(u):=u^TP_ru=\norm{P_ru}^2.
\]
The distance from $u$ to the great circle $S^3\cap\Pi_r$ is
\[
    d\bigl(u,S^3\cap\Pi_r\bigr)
       =\arccos\sqrt{x_r(u)}.
\]
Hence, if
\begin{equation}\label{eq:tstar}
    t_\star:=\min_{u\in S^3}\max_{1\leq r\leq18}x_r(u),
\end{equation}
then
\begin{equation}\label{eq:rho-tstar}
    \rho(\Uclosed)=\arccos\sqrt{t_\star}.
\end{equation}
Set
\[
    y_i(u):=u^TQ_i u,
    \qquad 1\leq i\leq9.
\]
By \eqref{eq:PfromQ},
\[
    x_i(u)=\frac12+y_i(u),
    \qquad
    x_{i+9}(u)=\frac12-y_i(u),
\]
and therefore
\begin{equation}\label{eq:maxx}
    \max_{1\leq r\leq18}x_r(u)
      =\frac12+\max_{1\leq i\leq9}|y_i(u)|.
\end{equation}
Now arrange the nine centered coordinates into a $3\times3$ matrix:
\begin{equation}\label{eq:Mdef}
M(u):=
\begin{pmatrix}
y_1&y_6&y_8\\
y_4&y_2&y_9\\
y_5&y_7&y_3
\end{pmatrix}.
\end{equation}
This is the key simplification.
\begin{proposition}[Quaternionic double cover]\label{prop:doublecover}
For every $u=(a,b,c,d)\in S^3$,
\begin{equation}\label{eq:Rformula}
R(u):=2M(u)=
\begin{pmatrix}
a^2+b^2-c^2-d^2 & 2(bc-ad) & 2(ac+bd)\\
2(ad+bc) & a^2-b^2+c^2-d^2 & 2(cd-ab)\\
2(bd-ac) & 2(ab+cd) & a^2-b^2-c^2+d^2
\end{pmatrix}
\in\SO(3).
\end{equation}
The map $u\mapsto R(u)$ is the quaternionic double cover: in particular,
\[
    R(uv)=R(u)R(v).
\]
It is surjective onto $\SO(3)$ and has fibers $\{u,-u\}$.
\end{proposition}
\begin{proof}
Substituting the nine matrices of Appendix~\ref{app:Q} into
$y_i=u^TQ_i u$ gives \eqref{eq:Rformula}.  The displayed matrix is the
standard rotation matrix for conjugation by the unit quaternion $u$: if
$v$ is a purely imaginary quaternion, then the coordinate vector of
$uvu^{-1}$ is $R(u)$ times the coordinate vector of $v$.  Therefore
$R(uv)=R(u)R(v)$.  The quaternionic conjugation action is the standard
double covering $S^3\to\SO(3)$, so it is surjective and its kernel is
$\{\pm1\}$.  Hence precisely $u$ and $-u$ give the same rotation.
\end{proof}
Combining \eqref{eq:maxx} with Proposition~\ref{prop:doublecover} gives an
exact reduction, not merely a bound:
\begin{equation}\label{eq:tstarSO3}
    t_\star
      =\frac12+\frac12
       \min_{R\in\SO(3)}\norm{R}_{\max},
    \qquad
    \norm{R}_{\max}:=\max_{i,j}|R_{ij}|.
\end{equation}
\section{The flattest rotation in three dimensions}
\label{sec:minimax}
The objective in \eqref{eq:tstarSO3} is the three-dimensional instance
of the flat-orthogonal-matrix problem studied by Jaming and Matolcsi
\cite{JamingMatolcsi2015}.  In their notation one considers
\[
    u_n:=\min_{M\in O(n)}\max_{i,j}|M_{ij}|.
\]
Changing the sign of one row switches the determinant without changing
the entry magnitudes, so the minima over $ O(n)$ and
$\SO(n)$ agree.  We prove that $ u_3=\frac{2}{3} $ and we provide a complete
equality classification.
\begin{theorem}\label{thm:minimax}
For every $R\in\SO(3)$,
\[
    \norm{R}_{\max}\geq\frac23.
\]
Equality is attained by
\begin{equation}\label{eq:Rstar}
    R_\star=\frac13
    \begin{pmatrix}
    -1&2&2\\
    2&-1&2\\
    2&2&-1
    \end{pmatrix}.
\end{equation}
Moreover every equality case is obtained from $R_\star$ by signed row
and column permutations that preserve determinant $+1$.
\end{theorem}
\begin{proof}
The displayed matrix $R_\star$ lies in $\SO(3)$ and has maximum entry
magnitude $2/3$, so only the lower bound requires proof.
Suppose first that $R=(r_{ij})\in\SO(3)$ satisfies
$\norm{R}_{\max}<2/3$.  Every entry must then have magnitude strictly
greater than $1/3$.  Indeed, if one coordinate of a unit row had
magnitude at most $1/3$, the other two squares would sum to at least
$8/9$, forcing one of them to have magnitude at least $2/3$.
Write
\[
R=
\begin{pmatrix}
a&b&c\\ d&e&f\\ g&h&i
\end{pmatrix}.
\]
Since $R^{-1}=R^T=\operatorname{adj}(R)$,
\begin{equation}\label{eq:cofactors}
    g=bf-ce,\qquad
    h=cd-af,\qquad
    i=ae-bd.
\end{equation}
Every product on the right has magnitude strictly between $1/9$ and
$4/9$.  If, for example, $ae$ and $bd$ had the same sign, then
\[
    |i|=\bigl||ae|-|bd|\bigr|<\frac13,
\]
contrary to the preceding paragraph.  Thus $ae$ and $bd$ have opposite
signs.  The other two cofactor identities similarly show that $bf$ and
$ce$ have opposite signs, and that $cd$ and $af$ have opposite signs.
Multiplying these three sign relations gives
\[
    \operatorname{sgn}(a^2b^2c^2d^2e^2f^2)=-1,
\]
which is impossible.  Hence $\norm{R}_{\max}\geq2/3$.
Now assume equality.  The same row-norm argument gives
$|r_{ij}|\geq1/3$ for every entry.  In the three cofactor identities
\eqref{eq:cofactors}, the three comparisons of signs cannot all be
opposite: the product of their three sign ratios is $+1$.  After signed
row and column permutations, we may therefore assume that $ae$ and $bd$
have the same sign.  Then
\[
    \frac13\leq|i|
       =\bigl||ae|-|bd|\bigr|
       \leq\frac49-\frac19=\frac13.
\]
Thus equality holds throughout, so one of $|ae|,|bd|$ equals $1/9$ and
the other equals $4/9$.  By applying, if necessary, a determinant-one
signed interchange of the first two columns, we may arrange
\[
    |a|=|e|=\frac13,
    \qquad
    |b|=|d|=\frac23.
\]
The unit lengths of the first two rows force
$|c|=|f|=2/3$.  The remaining cofactor identities then force
$|g|=|h|=2/3$ and $|i|=1/3$.  Hence, up to row and column permutations,
\[
    (|r_{ij}|)=\frac13
    \begin{pmatrix}
    1&2&2\\2&1&2\\2&2&1
    \end{pmatrix}.
\]
Finally use signed diagonal row and column operations to normalize the
first row to $(-1,2,2)/3$.  Orthogonality forces the second and third
rows, up to the remaining allowed signs, to be
$(2,-1,2)/3$ and $(2,2,-1)/3$.  The determinant condition fixes the last
sign.  Thus every equality case is a signed row/column permutation of
$R_\star$.
\end{proof}
\begin{corollary}[Number of extremal rotations]\label{cor:96}
Exactly $96$ matrices $R\in\SO(3)$ satisfy
$\norm{R}_{\max}=2/3$.
\end{corollary}
\begin{proof}
Let $B\leq O(3)$ be the signed permutation group, $|B|=48$, and let
\[
    G=\{(L,M)\in B\times B:\det(L)\det(M)=1\}.
\]
Then $|G|=48^2/2=1152$, and Theorem~\ref{thm:minimax} says that the
equality cases form the $G$-orbit of $R_\star$ under
$(L,M)\cdot R=LRM^{-1}$.
The stabilizer of $R_\star$ has order exactly $12$.  Indeed, in
$3|R_\star|$ the entries of magnitude $1$ occur precisely on the
diagonal.  Therefore the row permutation in a stabilizing pair uniquely
determines the column permutation, giving six possibilities.  Once these
permutations are fixed, comparison of the nonzero entries shows that the
row and column signs have exactly two simultaneous choices, differing by
$(-I,-I)$.  Hence
\[
    |\operatorname{Stab}_G(R_\star)|=6\cdot2=12,
\]
and orbit--stabilizer gives
\[
    \frac{1152}{12}=96.
\]
\end{proof}
\section{Proof of the main theorem and the deep holes}
\label{sec:deep-holes}
Theorem~\ref{thm:minimax} and \eqref{eq:tstarSO3} immediately give
\[
    t_\star
      =\frac12+\frac12\cdot\frac23
      =\frac56.
\]
Equation \eqref{eq:rho-tstar} therefore yields
\[
    \boxed{\rho(\Uclosed)=\arccos\sqrt{\frac56}}.
\]
An explicit deep hole is particularly simple.  Take
\begin{equation}\label{eq:ustar}
    u_\star=\frac1{\sqrt3}(0,1,1,1)\in S^3.
\end{equation}
Direct substitution in \eqref{eq:Rformula} gives
\[
    R(u_\star)=R_\star.
\]
Therefore
\[
    \max_r x_r(u_\star)
      =\frac12+\frac12\norm{R_\star}_{\max}
      =\frac56,
\]
so $u_\star$ is a deep hole.  Under our quaternion convention
$u=(a,b,c,d)\leftrightarrow aI-i(bX+cY+dZ)$, the corresponding antipodal
pair of unitaries is
\begin{equation}\label{eq:Ustar}
    \pm U_\star
      =\pm\frac1{\sqrt3}
       \begin{pmatrix}
       -i&-1-i\\
       1-i&i
       \end{pmatrix}.
\end{equation}
It remains to identify all the deep holes.  Let $B^+\leq\SO(3)$ be the
group of determinant-$+1$ signed permutation matrices.  It has order
$24$, and the restriction of $R$ to $\Cl_2$ maps onto $B^+$ with kernel
$\{\pm I\}$; this is the usual projective action of the Clifford group
on the three Pauli axes.
By Theorem~\ref{thm:minimax}, every extremal rotation can be written
$LR_\star M$, where $L$ and $M$ are signed permutation matrices with
$\det(L)\det(M)=1$.  If both determinants are $-1$, replacing $(L,M)$
by $(-L,-M)$ leaves $LR_\star M$ unchanged and makes both determinants
$+1$.
Thus the $96$ extremal rotations form a single left--right projective
Clifford orbit.  Each has two antipodal lifts in $S^3$, and multiplication
by $-I\in\Cl_2$ interchanges the two lifts.  Therefore all $192$ deep
holes form a single left--right binary-Clifford orbit.  This proves the
remaining assertions of Theorem~\ref{thm:main}.
Quantum gates are normally identified up to global phase.  Thus the
$192$ points in $\SU(2)$ represent $96$ projective gates.  Because
$\Uclosed$ itself is antipodally invariant, passing from the round metric
on $S^3$ to the induced projective metric does not change the covering
radius.
\section{Conclusion}
The closed single-qubit Clifford hierarchy has a surprisingly small
geometric model: eighteen great circles in $S^3$, supported on the
eighteen planes. This
turns the covering problem into the three-dimensional flat-orthogonal
matrix problem and gives the sharp all-level Clifford-fidelity threshold
$5/6$, together with all of its equality cases.
This suggests a natural higher-dimensional question.  Single-qudit
hierarchy gates in prime dimension are also semi-Clifford
\cite{deSilva2025clifford}, while their diagonal part is explicitly
classified \cite{CuiGottesmanKrishna}.  These results suggest analogous
covering problems for closures of finite Clifford translates of
higher-dimensional diagonal tori.  The especially simple identification
$\SU(2)/\{\pm1\}\cong\SO(3)$ is special to the qubit case, so new geometry
will be needed beyond dimension two.
\section*{Acknowledgments}
We thank Eric Kubischta for insightful discussions on the Clifford
hierarchy, and Jonas Anderson for a Quantum Computing Stack Exchange post
that inspired the geometric approach of this paper.
\appendix
\section{The nine centered projectors}
\label{app:Q}
All matrices are written in the orthonormal quaternion basis
$(1,\mathbf i,\mathbf j,\mathbf k)$.  Define
\[
Q_1=\frac12
\begin{pmatrix}
1&0&0&0\\0&1&0&0\\0&0&-1&0\\0&0&0&-1
\end{pmatrix},\quad
Q_2=\frac12
\begin{pmatrix}
1&0&0&0\\0&-1&0&0\\0&0&1&0\\0&0&0&-1
\end{pmatrix},\quad
Q_3=\frac12
\begin{pmatrix}
1&0&0&0\\0&-1&0&0\\0&0&-1&0\\0&0&0&1
\end{pmatrix},
\]
\[
Q_4=\frac12
\begin{pmatrix}
0&0&0&1\\0&0&1&0\\0&1&0&0\\1&0&0&0
\end{pmatrix},\quad
Q_5=\frac12
\begin{pmatrix}
0&0&-1&0\\0&0&0&1\\-1&0&0&0\\0&1&0&0
\end{pmatrix},\quad
Q_6=\frac12
\begin{pmatrix}
0&0&0&-1\\0&0&1&0\\0&1&0&0\\-1&0&0&0
\end{pmatrix},
\]
\[
Q_7=\frac12
\begin{pmatrix}
0&1&0&0\\1&0&0&0\\0&0&0&1\\0&0&1&0
\end{pmatrix},\quad
Q_8=\frac12
\begin{pmatrix}
0&0&1&0\\0&0&0&1\\1&0&0&0\\0&1&0&0
\end{pmatrix},\quad
Q_9=\frac12
\begin{pmatrix}
0&-1&0&0\\-1&0&0&0\\0&0&0&1\\0&0&1&0
\end{pmatrix}.
\]
The eighteen projectors are
\[
    P_i=\frac12I_4+Q_i,
    \qquad
    P_{i+9}=\frac12I_4-Q_i,
    \qquad i=1,\ldots,9.
\]
Thus the orthogonal-complement involution is visible directly as
$Q_i\leftrightarrow-Q_i$.

\end{document}